\documentclass[11pt]{article}

\date{}
\usepackage{float}
\usepackage[colorlinks,linkcolor=blue]{hyperref}
\usepackage[latin1]{inputenc}
\usepackage{amsmath}
\usepackage{graphicx}
\usepackage{amssymb}
\usepackage{amsthm}
\usepackage{verbatim}
\usepackage{epsfig}
\usepackage[labelfont=bf]{caption}
\usepackage{enumerate}
\usepackage{subfig}
\usepackage{epstopdf}

\usepackage{booktabs} 
\usepackage{algorithmicx,algorithm}

\newtheorem{definition}{Definition}
\newtheorem{assumption}{Assumption}

\newtheorem{remark}{Remark}

\newtheorem{proposition}{Proposition}

\def\begquo{\begin{quote}}
\def\endquo{\end{quote}}
\def\begequarr{\begin{eqnarray}}
\def\endequarr{\end{eqnarray}}
\def\begequarrs{\begin{eqnarray*}}
\def\endequarrs{\end{eqnarray*}}
\def\begarr{\begin{array}}
\def\endarr{\end{array}}
\def\begequ{\begin{equation}}
\def\endequ{\end{equation}}
\def\lab{\label}
\def\begdes{\begin{description}}
\def\enddes{\end{description}}
\def\begenu{\begin{enumerate}}
\def\begite{\begin{itemize}}
\def\endite{\end{itemize}}
\def\endenu{\end{enumerate}}

\def\lef[{\left[\begin{array}}
\def\rig]{\end{array}\right]}
\def\qed{\hfill$\Box \Box \Box$}
\def\begcen{\begin{center}}
\def\endcen{\end{center}}
\def\begrem{\begin{remark}\rm}
\def\endrem{\end{remark}}
\def\begdef{\begin{definition}}
\def\enddef{\end{definition}}
\def\begpro{\begin{proposition}}
\def\endpro{\end{proposition}}
\def\begfac{\begin{fact}}
\def\endfac{\end{fact}}
\def\begass{\begin{assumption}}
\def\endass{\end{assumption}}
\def\begsubequ{\begin{subequations}}
\def\endsubequ{\end{subequations}}

\def\begmat#1{\begin{bmatrix}#1\end{bmatrix}}
\def\begali#1{\begin{align}{#1}\end{align}}
\def\begalis#1{\begin{align*}{#1}\end{align*}}

\def\cali{{\cal I}}

\def\calc{{\cal C}}

\def\call{{\cal L}}

\def\caly{{\cal Y}}

\def\liminf{\lim_{t \to \infty}}

\def\L2e{{\cal L}_{2e}}

\def\rea{\mathbb{R}}

\def\adj{\mbox{adj}}
\def\col{\mbox{col}}
\def\tra{\mbox{tr}}
\def\hal{{1 \over 2}}

\def\JPC{{\it Journal of Process Control}}

\def\IJACSP{{\it Int. J. on Adaptive Control and Signal Processing}}

\def\TAC{{\it IEEE Trans. Automatic Control}}

\def\AUT{{\it Automatica}}

\def\CST{{\it IEEE Trans. Control Systems Technology}}

\usepackage{color}

\usepackage[prependcaption,colorinlistoftodos]{todonotes}

\begin{document}

\title{A Globally Stable Practically Implementable PI Passivity-based Controller for Switched Power Converters}
\author{Alexey Bobtsov$^{1}$, Romeo Ortega$^{1,2}$, Nikolay Nikolaev$^{1}$, Wei He*$^{3}$}
\maketitle

        {$^{1}$Department of Control
Systems and Robotics, ITMO University, Kronverkskiy av. 49, Saint
Petersburg, 197101, Russia, (e-mail: bobtsov@mail.ru, nikona@yandex.ru)}

        {$^{2}$Universit\'e Paris-Saclay, CNRS, CentraleSup\'elec,  Laboratoire des signaux et syst\`emes, 91190, Gif-sur-Yvette, France, (e-mail: ortega@lss.supelec.fr)}
        
        {$^{3}$School of Automation, Nanjing University of Information Science \& Technology, Nanjing, China, (e-mail: hwei@nuist.edu.cn)}

\begin{abstract}
In this paper we propose a  PI passivity-based controller, applicable to a large class of switched power converters, that ensures global state regulation to a desired equilibrium point. A  solution to this problem was reported in \cite{HERetal} but it requires {\em full state-feedback}, which makes it practically unfeasible. To overcome this limitation we construct a {\em state observer} that is implementable with measurements that are available in practical applications. The observer reconstructs the state in {\em finite-time}, ensuring global convergence of the PI. The excitation requirement for the observer is very weak and is satisfied in normal operation of the converters. Simulation results illustrate the excellent performance of the proposed PI.
\end{abstract}
%
\section{Introduction}
\label{sec1}
%
Switching power converters are, nowadays, an essential component of most electrical engineering applications. The ever increasing performance requirements on these applications translates into more stringent specifications on the quality of the converters control. The dynamics of power converters is highly nonlinear, even with fast switching when their averaged model is valid, and the validity of their linear approximation is restricted to a small neighborhood of the corresponding operating point \cite{KASSCHVERbook}.  The vast majority of power converters are controlled with classical PI loops and there is a widely accepted belief that, {\em if} they are properly tuned, their behavior is acceptable \cite{KAZetalbook,SIRSILbook,WANetalbook}. The qualifier ``if" in the previous sentence is very important because, if the range of operation of the system to be controlled is ``wide"---as it is required in modern applications---the task of tuning the gains of a PI (or, for that matter, of any other controller for nonlinear systems) is far from obvious. Unfortunately, to date, the only systematic procedure to carry-out this task, is invoking standard linear systems theory arguments, {\em e.g.}, pole location, stability margins, and applying them to the linearized model of the converter. Various procedures to re-tune the PI gains, including gain-scheduling \cite{VESILK}, relay auto-tuning \cite{ASTHAGbook} and adaptation \cite{ORTKEL}, have been proposed but they all suffer from, well-documented, serious limitations and drawbacks \cite{ASTMURbook,RUGSHA}.

The main objective of this paper is to propose a practically implementable PI controller, applicable to a large class of power converters described by average nonlinear models, which has the following properties.
\begenu[{\bf F1}]
\item {\em Global stability} of the closed-loop is guaranteed {\em for all} positive values of the PI tuning gains.
\item By ``practically implementable" we mean that it {\em does not} assume the availability of the full state of the system, and it only requires the measurements used in standard PIs.
\item The proposed PI has a very close connection with the widely popular PQ instantaneous power controllers of \cite{AKAbook}.
\endenu
Providing a theoretical framework for controller design  is a topic of paramount importance. On one hand, it allows power engineer practitioners to apply the control law with confidence and, on the other hand, considerably simplifies the commissioning stage, which now that stability is guaranteed for all positive tuning gains, can concentrate on the transient performance specifications.

The approach adopted in this paper is in the line of \cite{ORTetalcsm} which relies on the use of energy balance concepts to control a system. Unlike most classical nonlinear control techniques found in the literature, which try to impose some predetermined dynamic behavior---usually through nonlinearity cancellation, domination of nonlinearites and high gain---energy-based methods exploit and respect the physical structure of the system. Passivity-based control (PBC) is the generic name of this controller design methodology, which achieves stabilization by exploiting the passivity properties of the system. Due to the physical appealing that this methodology has, a vast literature has been advocated to its application to mechanical, electrical and electromechanical systems, see \cite{CIS,ORTetalbook}.

Passivity is the key property of power converters that is exploited in this paper.  The first time that passivity principles were applied for power converter control is  in the foundational paper \cite{SANVER}. It is well-known \cite{ORTetalbook,VANbook} that wrapping a PI (or a PID) around a passive output ensures input-output stability of the system and convergence to zero of the passive output. In many applications, including power converters, the objective is to drive this output to a value {\em different} from zero, which is associated with the steady-state behavior of the system. In this case, we are interested in proving that the incremental model is passive. In \cite{PERORTESP} it was shown that, for  a general class of models of power converters, it is possible to define an output signal such that the incremental model is passive---this result was later generalized in several directions in \cite{HERetal,HERetalbook}. This fundamental property, called ``passivity of the nonlinear incremental model" in \cite{JAYetal}, is now  referred as {\em shifted passivity} \cite{VANbook} and has played a central role in many recent developments of the control community. In the present context the main interest of this property is that, driving the shifted passive output to zero with a PI control ensures, under a precise condition on the systems dissipation, that the {\em full state} of the system converges to a desired equilibrium.

The  PI controller of  \cite{HERetal} has enjoyed a very wide popularity---with more than 120 cites in Google Scholar. However, its main drawback is that it requires the {\em measurement of the state} that, for technological and economic reasons, is unattainable in most practical applications of power converters. The {\em main contribution} of this note is to propose an observer that, using standard measurements, reconstructs the system state in {\em finite-time} giving rise to a, practically reasonable, globally stable PI design. Two additional features of the observer are, on one hand, that convergence of the state estimates is achieved under very weak excitation assumptions usually met in normal power converter operation and, on the other hand, that it does not rely on noise-sensitive high-gain injection or approximate differentiation.

The design of the observer relies on the use of the following two theoretical developments recently reported by the authors.

\begenu[{\bf T1}]
\item Parameter Estimation-based Observer (PEBO) design technique first reported in \cite{ORTetalscl}, whose main novelty is that it translates the task of state observation into an on-line  \emph{parameter estimation} problem. In this work we use a generalization of this design, called GPEBO, recently used in \cite{ORTetaljpc} for the state observation of chemical-biological reactors.
\item To estimate the parameters required by GPEBO, we use the Dynamic Regressor Extension and Mixing (DREM) procedure first proposed in \cite{ARAetaltac}. The central property of DREM is that it generates, out of a $q$-dimensional Linear Regression Equation (LRE), $q$ {\em scalar} LREs---one for each of the unknown parameters---considerably simplifying the estimation task and ensuring convergence under very weak excitation assumptions.
\endenu

The remainder of the paper is organized as follows. In Section \ref{sec2} we present the mathematical model of the power converters we consider in the paper, recall the state-feedback PI-PBC of \cite{HERetal} and formulate the observer problem, for which we give a solution in Section \ref{sec3}.  Simulation results, which illustrate the {performance} of the proposed PI-PBC, are presented in Section \ref{sec4}. The paper is wrapped-up with concluding remarks and future research in  Section \ref{sec5}.\\

\noindent {\bf Notation.} $I_n$ is the $n \times n$ identity matrix.  For $x \in \rea^n$,  we denote the Euclidean norm $|x|^2:=x^\top x$. All mappings are assumed smooth. Given a function $f:  \rea^n \to \rea$ we define the differential operator $\nabla f:=\left(\frac{\displaystyle \partial f }{\displaystyle \partial x}\right)^\top$.
%
\section{The PI-PBC of \cite{HERetal} and the Observer Problem Formulation}
\label{sec2}
%
We consider in this paper  switched power converters, with switched external sources, described in port-Hamiltonian (pH) form \cite{ESCVANORT,VANbook}
\begin{align}
\label{sys}
\dot x=\left(J_0+\sum^m_{i=1}J_iu_i-R\right)\nabla H(x)+\left(G_0+\sum^m_{i=1}G_iu_i\right)E,
\end{align}
where $x(t) \in \mathbb{R}^n$, consisting of inductor fluxes and capacitor charges, is the converter state, which is assumed {\em unknown}, $u(t) \in \mathbb{R}^m$ is the duty ratio, $E \in \mathbb{R}^n$ is the vector of constant external sources, $J_i=-J^\top _i \in \mathbb{R}^{n \times n},\; i \in \bar m:=\{0, \cdots, m\}$, are the interconnection matrices, $R =R^\top > 0$ denotes the dissipation matrix, $H:\rea^n \to \rea_+$ is the total stored energy, and $G_i \in \mathbb{R}^{n \times n},\; i \in \bar m$, are the input matrices. Assuming, linear capacitors and inductors, we have that
$$
H(x)=\frac{1}{2}x^\top Qx,
$$
with $Q \in \rea^{n \times n},\;Q>0$, a diagonal matrix. We assume that all the parameters of the converter are {\em known}. The model \eqref{sys} describes the behavior of most power converters used in practice \cite{KASSCHVERbook,KAZetalbook,SIRSILbook,WANetalbook}.\\

\noindent {\bf Control Objective} Fix an admissible equilibrium point ${x^\star}   \in \mathbb{R}^n$, that is, ${x^\star} $ satisfies
       \begin{eqnarray}
\label{addequ}  0 &=& \left ( J_{0}+\sum_{i=1}^m J_{i} u_{i}^{*} - R \right ) Q{x^\star}  + \left ( G_{0} + \sum_{i=1}^m G_{i} u_{i}^\star  \right ) E,
       \end{eqnarray}
for some $u^\star  \in \mathbb{R}^m$. Design a controller such that {\em global state regulation} is achieved, namely, that for all system and controller initial conditions,
\begequ
\lab{stareg}
\liminf x(t)={x^\star}  ,
\endequ
with all signals bounded.
\subsection{The full-state feedback PI-PBC of \cite{HERetal}}
\label{subsec21}
%
The following {\em full-state} feedback PI solution to the aforementioned problem was proposed in \cite{HERetal}.
%
\begin{proposition}\em
\label{pro1}
Consider the switched power converter described by (\ref{sys}) in closed loop with the PI-PBC
\begin{eqnarray}
\nonumber
\dot x_c & = & \tilde{y}\\
u & = & -K_P \tilde{y} - K_I x_c,
\label{pi}
\end{eqnarray}
with $K_P=K_P^\top\geq0,\;,K_I=K_I^\top\geq0$ and
\begequ
\lab{tily}
\tilde y = \calc (x-{x^\star} ),
\endequ
with $\calc \in \rea^{m \times n}$ given by
\begequ
\lab{calc}
\calc = \lef[{c}  E^\top G_1^\top-({x^\star}  )^\top Q J_1  \\ \vdots \\ E^\top G_m^\top-({x^\star}  )^\top Q J_m  \rig]Q=:G_N^\top (x^\star  )Q,
\endequ
For all initial conditions $(x(0),x_c(0)) \in  \mathbb{R}^{n} \times \rea^ m$ the trajectories of the closed-loop system are bounded and satisfy \eqref{stareg}.
\qed
\end{proposition}

The proof of Proposition \ref{pro1} is established in \cite{HERetal} showing that the system \eqref{sys} in closed-loop with the PI-PBC \eqref{pi} is described by
\begequ
\lab{dottilx}
\dot {\tilde x}=\left(J_0+\sum^m_{i=1}J_iu^\star  _i-R\right)Q \tilde x+G_N(x)\tilde u,
\endequ
where the matrix $G_N(\cdot)$ is defined in \eqref{calc} and we (generically) use the notation $\tilde {(\cdot)}:=(\cdot)-{(\cdot)^\star}$. Therefore, the Lyapunov function candidate
$$
W(\tilde x, \tilde x_c)= \hal \tilde x^\top  Q \tilde x + \hal \tilde x_c^\top  K_I \tilde x_c
$$
where $x_c^\star:=K_I^{-1} u^\star$, verifies
$$
\dot W = -\tilde x^\top  QRQ \tilde x -  \tilde y^\top  K_P \tilde y.
$$
In \cite{HERetal} it is assumed that the dissipation matrix $R$ is only positive {\em semidefinite}. Consequently, to ensure asymptotic stability, it is necessary to impose a detectability condition that, as shown in \cite{CIS,HERetal,HERetalbook,MEZetal}, is satisfied in practical converter topologies. For ease of presentation, we strengthen here the assumption on $R$, which is tantamount to making the physically reasonable assumption that all energy storage elements---that is, inductors and capacitors---are leaky.

\begrem
\lab{rem1}
In \cite{PERORTESP} it is shown that. for a classical rectifier, one of the components of the passive output $\tilde y$, defined in \eqref{tily}, is related with the difference between supplied and extracted instantaneous active powers (of a suitably scaled representation) of the rectifier, the second output being the component of the input current. In this way, we make a nice connection with the widely popular {\em PQ instantaneous power} controllers of \cite{AKAbook}, where an outer PI loop around the output voltage is used to generate a reference for an inner PI loop acting on the aforementioned power difference. It is also possible to show that driving the passive output  to zero can be reinterpreted as a power equalization objective identical to the one used in Akagi's PQ method.
\endrem
\subsection{State observation problem formulation}
\label{subsec22}
%
The main objective of this paper is to design a globally convergent {\em state observer} that, combined with the PI-PBC of Proposition \ref{pro1}, will give rise to a practically implementable version of it. To pose the observer design problem we assume that there exists a full-rank matrix $C \in \mathbb{R}^{p \times n},\;p<n$, such that the vector
\begequ
\lab{ym}
y_m=Cx,
\endequ
is available for measurement. Typically, this signal will be the voltage fed to the converter, the input current flowing into it and/or the voltage provided to the load  \cite{KASSCHVERbook}. It will be shown below that---as expected---the level of excitation required to ensure convergence of the observer will reduce if the number of measurements increases---see Remark \ref{rem3}.\\

\noindent {\bf Finite-convergence Time (FCT) Observer Problem Formulation} Consider the power converter \eqref{sys} with {\em known} parameters and measurable output signals \eqref{ym}. Design a dynamical system
\begalis{
\dot{\chi} & =  F(y_m, \chi)\\
\hat x & = H(y_m,\chi)
}
with $\chi \in \rea^{n_{\chi}}$, such that, for all initial conditions $(x(0),\chi(0)) \in \mathbb{R}^{n + n_\chi}$, all signals remain bounded and there exists $t_c \in [0,\infty)$ such that
$$
\hat x(t) = x(t),\; \forall t \geq t_c.
$$

\begrem
\lab{rem2}
Following standard practice in observer theory  \cite{BERbook} we assume that  the the control $u$ is such that the state trajectories of \eqref{sys} are bounded.
\endrem
%
\section{Observer Design}
\label{sec3}
%
In this section we present the main result of the paper, namely an FCT GPEBO+DREM to reconstruct the state of the system \eqref{sys} with the measurable outputs \eqref{ym}. The excitation assumption for {\em global asymptotic} convergence is very week and, as illustrated in the simulations in Section \ref{sec5}, is satisfied in normal operating conditions of the converters. We also present other asymptotically convergent observers, which have a below par performance with respect to the GPEBO.
\subsection{GPEBO+DREM design with FCT}
\label{subsec31}
%
To streamline the presentation of the result we define the matrix
\begequ
\lab{a}
\Lambda(u):=\left(J_0+\sum^m_{i=1}J_iu_i-R\right)Q,
\endequ
and make the following sufficient excitation condition \cite{KRERIE}.
\begin{assumption}\label{ass1}\em
Fix a small constant $\mu \in (0, 1)$. There is a time $t_c>0$ such that
\begequ
\lab{sufexccon}
\int^{t_c}_0 \Delta^2(\tau)d\tau\geq-\frac{1}{\gamma}\ln(1-\mu).
\endequ
\end{assumption}

\begin{proposition}
\label{pro2} \em
Consider the system \eqref{sys} with the measurable outputs \eqref{ym}. Define the GPEBO via
\begsubequ
\lab{gpebo}
\begali{
\dot \xi & =\Lambda(u)\xi+\left(G_0+\sum^m_{i=1}G_iu_i\right)E\label{dotxi}\\
\dot \Phi &=\Lambda(u)\Phi,\;\Phi(0)=I_n \lab{dotphi}\\
\dot Y&=-\lambda Y + \lambda \Phi^\top  C^\top (C\xi-y_m) \lab{doty}\\
\dot \Omega&=-\lambda \Omega + \lambda \Phi^\top  C^\top  C\Phi \lab{dotome}\\
\dot\omega &=-\gamma\Delta^2\omega,\; \omega(0)=1 \lab{dotw}\\
\dot {\hat \theta}&=\gamma\Delta(\mathcal{Y}-\Delta\hat \theta),\label{dothatthe}
}
\endsubequ
with $\lambda>0,\; \gamma>0$, free tuning parameters and
\begali{
\mathcal{Y}=\text{adj}\{\Omega\}Y,\;\Delta =\det\{\Omega\},
\lab{calydel}
}
where $\adj\{\cdot\}$ is the adjunct  matrix. The state estimate
\begequ
\lab{hatxfct}
\hat x=\xi+\Phi\hat\theta_{\tt FCT}
\endequ
where we introduced the vector
\begequ
\lab{hatthefct}
\hat\theta_{\tt FCT}:=\frac{1}{1-\omega_c}\left[\hat \theta-\omega_c\hat\theta(0)\right]
\endequ
and $\omega_c$ is defined via the clipping function
\begin{align}\nonumber
\omega_c=
\begin{cases}
\omega~~~~~~~\text{if}~~\omega<1-\mu,\\
1-\mu~~\text{if}~~\omega\geq1-\mu.
\end{cases}
\end{align}
ensures that, for all initial conditions $(\xi(0),Y(0),\Omega(0),\hat \theta(0)) \in \rea^{n} \times \rea^{n} \times \rea^{n \times n} \times \rea^{n}$, we have that
\begequ
\lab{hatequx}
\hat x(t)=x(t),~~\forall t>t_c
\endequ
with all signals bounded provided $\Delta$ verifies Assumption \ref{ass1}.
\end{proposition}

\begin{proof}
Define the signal  $e:=x-\xi.$ From \eqref{sys} and \eqref{dotxi} one gets the linear time-varying (LTV) system
$$
\dot e = A(t) e,
$$
where we have defined $A(t):=\Lambda(u(t))$. The state transition matrix of this LTV system satisfies the equation \eqref{dotphi} \cite{RUGbook}, consequently
$$
e:=\Phi\theta,
$$
where $\theta:=e(0)$ is an {\em unknown} constant vector. Using the definition of $e$ we get
\begequ
\lab{xxiphi}
x=\xi+\Phi\theta.
\endequ
Next, our task is to estimate the parameter $\theta$ so that we can reconstruct the state $x$ using \eqref{xxiphi}. Towards this end, we use the output measurements \eqref{ym} to  get
$$
y_m=Cx=C\xi+C\Phi\theta,
$$
which we can write as a LRE
\begequ
\lab{ymmincphi}
y_m - C\xi=C\Phi\theta.
\endequ
Following the DREM procedure \cite{ARAetaltac} we carry out the next operations
	\begalis{
	\Phi^\top  C^\top (y_m - C\xi) &= \Phi^\top  C^\top  C\Phi  \theta \qquad \qquad \quad\;(\Leftarrow\; \Phi C^\top \times \eqref{ymmincphi})\\
		{\lambda \over p + \lambda}[\Phi^\top  C^\top (y_m - C\xi) ] & = {\lambda \over p + \lambda}[ \Phi^\top  C^\top  C\Phi]  \theta \qquad (\Leftarrow\;{\lambda \over p + \lambda}[\cdot])\\
		Y &= \Omega \theta \qquad \qquad \;\qquad \qquad\;(\Leftrightarrow\;\eqref{doty}, \eqref{dotome})\\
		\adj\{\Omega\} Y &= \adj\{\Omega\} \Omega \theta \qquad \qquad\quad\; (\Leftarrow\;\adj\{\Omega\} \times)\\
		\caly &= \Delta \theta  \qquad \qquad\qquad \quad \quad\;(\Leftrightarrow\;\eqref{calydel}).
	}
where, to obtain the last identity, we have used the fact that for any (possibly singular) $n \times n$ matrix $M$ we have $\adj\{M\} M=\det\{M\}I_n$. Replacing the latter equation in \eqref{dothatthe} yields the error dynamics
	\begequ
	\lab{parerrequ}
	\dot {\tilde \theta}=-\gamma\Delta^2  \tilde \theta,
	\endequ
where $ {\tilde \theta}:=\hat \theta - \theta$.  Since $\Delta$ is a {\em scalar}, the solution of the latter equation is given by
	\begequ
	\lab{tilthe}
	\tilde \theta=e^{-\gamma \int_0^t \Delta^2(s)ds}\tilde \theta(0),\;\forall t \geq 0.
	\endequ
Now, notice that the solution of \eqref{dotw} is
$$
w(t)=e^{-\gamma \int_0^{t}\Delta^2(s)ds}.
$$
The key observation is that, using the equation above in \eqref{tilthe}, and rearranging terms we get that
\begequ
\lab{keyrel}
[1-w(t)]\theta= \hat \theta(t) - w(t) \hat \theta(0).
\endequ
Finally, observe that $w$ is a non-increasing function and, under the interval excitation Assumption \ref{ass1}, we have that
$$
w_{c}(t)=w(t) < \mu,\;\forall t \geq t_c,
$$
The proof of state estimation convergence is completed noting that the latter implies that
$$
\frac{1}{1-\omega_c(t)}\left[\hat \theta(t)-\omega_c(t)\hat\theta(0)\right]=\theta,\;\forall t \geq t_c,
$$
that, in view of \eqref{hatxfct} and \eqref{xxiphi}, implies \eqref{hatequx}.

We proceed now to prove that all signals of the GPEBO are bounded. First, notice that \eqref{dotxi} is a copy of the converter dynamics \eqref{sys}. Therefore, in view of Remark \ref{rem2}, $\xi$ is bounded. To prove boundedness of $\Phi$, consider the mapping $U(\Phi)=\hal \tra\{\Phi^\top  Q \Phi\}$ whose derivative, along the trajectories of \eqref{dotphi} is
\begali{
\nonumber
\dot U & = \sum_{i=1}^n \Phi^\top _i Q \Lambda(u) \Phi_i\\
\lab{dotu}
& = \sum_{i=1}^n \Phi^\top _i Q R Q \Phi_i \leq  - c U,
}
where $\tra\{\cdot\}$, denotes the trace, $\Phi_i \in \rea^n$ is the $i$-th column of the matrix $\Phi$ and $c>0$.  Now, from \eqref{dotome} we se that boundedness of $\Phi$ implies boundedness of $\Omega$. Finally, the fact that $Y$ is bounded, follows from the fact that $Y=\Omega \theta$. This completes the proof.
\end{proof}

\begrem
\lab{rem3}
We make the observation that, as the number of measurements $p$ increases, the sufficient excitation Assumption \ref{ass1} is ``easier" to satisfy. In the extreme case when $p=n$ the matrix $\Omega$ is positive definite, hence $\Delta$ is positive and bounded away from zero. In any case, given the monotonicity of $\Delta$, the sufficient excitation condition of Assumption \ref{ass1} is ``almost always" satisfied.
\endrem

\begrem
\lab{rem4}
In the time interval $[0,t_c]$ the FCT estimated parameter \eqref{hatthefct} takes the form
$$
\hat\theta_{\tt FCT}:=\frac{1}{\mu}\left[\hat \theta-(1 - \mu)\hat\theta(0)\right].
$$
Therefore, if we choose the constant $\mu$ close to zero there is a potential {\em high-gain} injection. On the other hand, since the right hand side of \eqref{sufexccon} becomes smaller, this ``reduces" the size of $t_c$. The solution of this compromise is further complicated by the fact that the adaptation gain $\gamma$ also appears into the picture. Unfortunately, this conundrum makes the tuning stage of the estimator quite complicated---a fact that is observed in the simulations of Section \ref{sec5}, where the choices of the constants $\lambda,\;\gamma$ and $\mu$ are done via trial-and-error.
\endrem
\subsection{Other globally asymptotically convergent observers}
\label{subsec32}
%
In this subsection we present three alternative observer designs that, alas, exhibit some drawbacks, which are conspicuous by their absence from the FCT-GPEBO of Proposition \ref{pro2}.
\subsubsection*{An open-loop emulator}
In the proof of Proposition \ref{pro2} it was shown that the state transition matrix converges, exponentially fast, to zero, see \eqref{dotu}. Consequently, the simplest way to reconstruct the system state is by setting $\hat x=\xi$, which from \eqref{xxiphi} and  \eqref{dotu} has the property that $\hat x(t) \to x(t)$, (exp). This type of constructions, known as {\em emulators}, are the dominant technique in process control \cite{DOC}, where they are called ``asymptotic observers".  There are two main drawbacks to this approach: first, the convergence rate is fixed by the system dynamics, in this case by the smallest eigenvalue of the matrix $Q^\top RQ$, and cannot be tuned. Second, being an open-loop construction, it suffers from serious robustness problems, being highly sensitive to parameter uncertainty and noise. An alternative to overcome these drawbacks in process control applications is precisely GPEBO, and it was reported in \cite{ORTetaljpc}.
\subsubsection*{A Kalman-Bucy filter}
Evaluating $u$ along trajectories, as done in the proof of Proposition \ref{pro2}, it is possible to view the system \eqref{sys} and \eqref{ym} as an LTV system
\begalis{
\dot x&=A(t) x + B(t)\\
y_m&=Cx,
}
with $B(t):=\left(G_0+\sum^m_{i=1}G_iu_i(t)\right)E$. Then, we can design a classical Kalman-Bucy filter (KBF) for it \cite{AND}
\begalis{
\dot{\hat  x} & = [A(t)-HC^\top C]\hat x +B(t)\\
\dot H &= HA^\top (t)+A(t)H-HC^\top C H+S,
}
with $H(0)>0,\;S>0$, which ensures $\liminf |\hat x(t)-x(t)|=0$ (exp), provided the pair $(A(t),C)$ is uniformly completely observable, that is, if it satisfies
\begequ
\lab{uco}
c_1 \cali_n \geq \int_{t_0}^{t_0+T}\Phi^\top (\tau,t_0)C^\top C \Phi(\tau,t_0)d\tau \geq c_2 \cali_n,\;\forall t_0 \geq 0,
\endequ
for some positive constants $c_1, c_2$ and $T$---see \cite{AND}. Alas, $\Phi(t) \to 0$ (exp) for the power converter system, complicating the task of satisfying the lower bound on the inequality.
\subsubsection*{GPEBO with classical gradient estimators}
From the LRE \eqref{ymmincphi} we can propose a classical gradient estimator of $\theta$, that is
$$
\dot {\hat \theta}=\gamma\Phi^\top  C^\top (y_m-C\xi-C\Phi\hat \theta),
$$
whose error equation is given by
$$
\dot {\tilde \theta}=-\gamma\Phi^\top  C^\top  C\Phi\tilde \theta.
$$
It is well-known \cite{SASBOD} that this error equation is globally exponentially stable if and only $C\Phi$ is PE---that is, if and only if \eqref{uco} is satisfied. As discussed above, this condition is hard to satisfy in our application.

A gradient estimator can also be applied to the LRE $Y=\Omega \theta$, obtained with the dynamic regressor extension \eqref{doty}, \eqref{dotome}. The PE requirement is now imposed to the matrix $\Omega$, which is a far more stringent condition than Assumption \ref{ass1}.

Another possibility is to apply the concurrent learning algorithms proposed in \cite{CHOetal,KAMetal} that require only sufficient excitation of the regressor $\Omega$. This is, again, a more stringent requirement than Assumption \ref{ass1}.
%
\section{Observer-based PI-PBC}
\lab{sec4}
%
In this section we propose our output-feedback PI-PBC, which is obtained combining the PI controller of Proposition \ref{pro1} with the state observer reported in the previous section.

\begin{proposition}\em
\lab{pro3}
Consider the system \eqref{sys} with the measurable outputs \eqref{ym} in closed-loop with the PI controller of Proposition \ref{pro1}, where the state $x$ is replaced by an estimate $\hat x$  generated via the FCT GPEBO of Proposition \ref{pro2} and Assumption \ref{ass1} holds.

Then, {\em for all} initial conditions of the plant and the controller $(x(0),x_c(0)) \in \mathbb{R}^{n} \times \rea^ m$ and {\em all} initial conditions of the observer $(\xi(0),Y(0),\Omega(0),\hat \theta(0)) \in \rea^{n} \times \rea^{n} \times \rea^{n \times n} \times \rea^{n}$, the trajectories of the closed-loop system are bounded and satisfy \eqref{stareg}.
\end{proposition}

\begin{proof}
First, notice that the output signal \eqref{tily}, evaluated with $\hat x$ instead of $x$, may be written as
\begalis{
\calc (\hat x - {x^\star}   ) &= \calc[x - {x^\star}    + \Phi (\hat\theta_{\tt FCT}-\theta)]\\
& = \tilde y + \calc \Phi  (\hat\theta_{\tt FCT}-\theta).
}
Consequently, the PI-PBC \eqref{pi} takes the form
\begalis{
\dot x_c &=\tilde y + \calc \Phi  (\hat\theta_{\tt FCT}-\theta) =:F_{x_c}(x) + \calc \Phi  (\hat\theta_{\tt FCT}-\theta) \\
\tilde u &= -K_P \tilde y -K_I \tilde x_c - K_P \calc \Phi  (\hat\theta_{\tt FCT}-\theta).
}
Replacing the control signal $\tilde u$ above in the system dynamics \eqref{dottilx} yields
\begalis{
\dot {\tilde x} &=\left(J_0+\sum^m_{i=1}J_iu^\star  _i-R\right)Q \tilde x-G_N(x)[K_P \tilde y +K_I x_c + K_P \calc \Phi  (\hat\theta_{\tt FCT}-\theta)] \\
&=:F_x(x,x_c)-G_N(x)K_P \calc \Phi (\hat\theta_{\tt FCT}-\theta) ,
}
where $F_x(x,x_c)$ corresponds to the system in closed-loop with the full-state measurable input signal $\tilde u$.

Proposition \ref{pro1} ensures that $({x^\star},x_c^\star )$ is an asymptotically stable equilibrium of the unperturbed system.  Now, as shown in Proposition \ref{pro2}, Assumption \ref{ass1} ensures that $\hat\theta_{\tt FCT}(t)=\theta$ for all $t \geq t_c$. Invoking the absence of finite-escape times and the aforementioned asymptotic stability allows us to conclude that \eqref{stareg} holds globally.
\end{proof}

\begrem
\lab{rem5}
From \eqref{tilthe} it is clear that, under the assumption that $\Delta \notin \call_2$, it is possible to globally, asymptotically reconstruct the state $x$ with the simpler state estimate $\hat x=\xi+\Phi\hat\theta$. However, proving that the resulting certainty equivalent observer-based PI-PBC is globally stable seems very elusive. There are two technical reasons for this situation: first, that neither for the $(\tilde x, \tilde x_c)$ subsystem nor for the $\tilde \theta$  subsystem, we dispose of exponential stability proofs, which are needed to dominate---with suitable negative definite quadratic terms---the sign-indefinite cross-terms appearing in the Lyapunov-based analysis. Second, the perturbation term of the $\tilde x$ dynamics, although linear in $\tilde \theta$, is multiplied by the factor $G_N(x)$, whose boundedness is yet to be proven. On the other hand, invoking \cite[Theorem 3.1]{VID}, it is possible to establish a local stability result simply introducing a normalization to the DREM estimator.
\endrem
\section{Application to \'{C}uk Converter}
\lab{sec5}
%
In this section we apply the results of the paper to design of an output-feedback PI for the \'{C}uk converter assuming measurable the output voltage. For this example, it is possible to design an adaptive version of it, where we treat the resistive load as an {\em unknown} parameter.
\subsection{Control problem formulation}
\lab{subsec51}
The average model of the \'{C}uk converter is given by \cite{ASTKARORTbook,ORTetalbook}
\begin{align*}
L_1\dot i_1=&-r_1i_1-(1-u)v_2+E,\nonumber\\
C_1\dot v_2=&(1-u)i_1+ui_3,\nonumber\\
L_2\dot i_3=&-r_2i_3-u v_2-v_4,\nonumber\\
C_2\dot v_4=&i_3-\frac{v_4}{r},
\end{align*}
where $(i_1,v_2,i_3,v_4)$ is the state vector and $u \in (0,1)$ is the duty ratio of the transistor switch, with the signals and all the constants shown in Fig. \ref{fig1}. The model can be rewritten in the form \eqref{sys} with $x=(L_1i_1,C_1v_2,L_2i_3,C_2v_4)$, $m=1$ and the matrices
\begin{align*}
J_0=&\left[
      \begin{array}{cccc}
        0 & -1 & 0 & 0 \\
        1 & 0 & 0 & 0 \\
        0 & 0 & 0 & -1 \\
        0 & 0 & 1 & 0 \\
      \end{array}
    \right],\;
J_1=\left[
      \begin{array}{cccc}
        0 & 1 & 0 & 0 \\
        -1 & 0 & 1 & 0 \\
        0 & -1 & 0 & 0 \\
        0 & 0 & 0 & 0 \\
      \end{array}
    \right],\;G_0=\begmat{ 1 \\ 0 \\ 0 \\ 0},\;G_1=0 \\
R=&\text{diag}\Big\{r_1, 0, r_2, \frac{1}{r}\Big\},\;Q=\text{diag}\Big\{\frac{1}{L_1}, \frac{1}{C_1}, \frac{1}{L_2}, \frac{1}{C_2}\Big\}.
\end{align*}
The control objective is to drive the output voltage $x_4$ towards a desired valued $x_4^\star  <0$, assuming only partial state measurements.

\begin{figure}[h]
\centering
\includegraphics[width=0.7\linewidth]{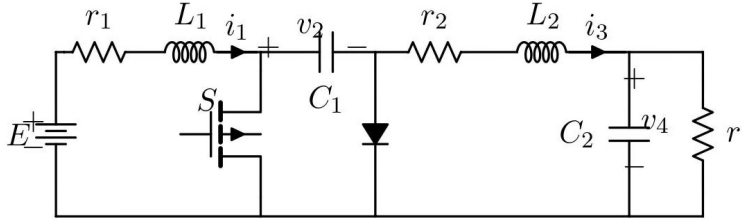}
	\caption{DC-DC \'{C}uk converter circuit}
	\label{fig1}
\end{figure}

\subsection{State observation problem}
\lab{subsec52}
The problem of state observation of this system was solved using PEBO in \cite{ORTetalscl} assuming measurable either $(x_2,x_4)$ or $(x_2,x_3)$. For the latter measurements an adaptive observer and an output feedback controller were designed in \cite{ASTKARORTbook} using the immersion and invariance technique. To show the benefit of the present GPEBO approach we assume here that {\em only} $x_4$ is available for measurement---which is a physically reasonable assumption. Hence, in \eqref{ym} we have $C=\left[ \begin{array}{cccc} 0 & 0 & 0 & 1 \\ \end{array} \right].$

Given the definitions above, the observer of Proposition \ref{pro2} can be designed with
\begin{align*}
\nonumber
\Lambda(u)=
\begin{bmatrix}
-\dfrac{r_1}{L_1}& -\dfrac{1-u}{L_1}& 0& 0&  \\
\dfrac{1-u}{C_1}& 0& \dfrac{u}{C_1}& 0& \\
0& -\dfrac{u}{L_2}& -\dfrac{r_2}{L_2}& -\dfrac{1}{L_2}&  \\
0& 0& \dfrac{1}{C_2}& 0&
\end{bmatrix}.
\end{align*}

\subsection{Design of the PI-PBC of Proposition \ref{pro1}}
\lab{subsec53}
Some lengthy, but straightforward, calculations show that, fixing $x_4=x_4^\star $, the equilibria of the state can be parameterized as
\begin{align*}
x^\star = \begin{bmatrix}
{-\frac{u^\star  }{r(1-u^\star )}}\\
-\frac{1}{u^\star } (1+\frac{r_2 }{r})\\
\frac{1 }{r}\\
1
\end{bmatrix}x_4^\star =:d(u^\star ) x_4^\star
\end{align*}
with the equilibrium value for the control $u^\star $ given as
\begin{align*}
u^\star =\frac{-a_1+\sqrt{a_1^2+4 a_2 a_0}}{2 a_2},
\end{align*}
where we defined the constants
$$
a_0:=(r +r_2) x^\star _4,\;a_1:=E r - 2(r +r_2) x^\star _4,\;a_2 := (r_1+ r +r_2) x^\star _4 - Er.
$$

Given these definitions it is possible to define the incrementally passive output $\tilde y$ given in \eqref{tily} as
\begin{align*}
\tilde{y} = -(x^\star)^\top Q J_1 Q x={x_1^\star  \over L_1}  {x_2\over C_1} + {x_2^\star   \over C_1}\Big({x_3 \over L_2} - {x_1\over L_1}\Big) - {x_3^\star  \over L_2}{ x_2 \over C_1}.
\end{align*}
Notice that this signal does not depend on $x_4$.

\subsection{Simulation results}
\lab{subsec54}
In all the simulations presented in this section the converter was put in closed-loop with the output-feedback PI-PBC of Proposition \ref{pro3} with the gains $K_P=10$ and $K_I=5$, and the desired output voltage fixed at $x_4^\star =-20$.  The circuit parameters are shown in Table \ref{tab:simpars}. In all simulations, the control signal was saturated to verify that it remains in the valid interval (0,1). For the sake of brevity the plots are omitted.

\begin{table}[!ht]
\caption{Circuit parameters of the \'{C}uk converter.}
\centering
\begin{tabular}{lccc}
\hline
~Parameter~&Symbol (unit)~~& \multicolumn{2}{c}{Value} \\
\hline
~Input voltage~~ & $E~(\text{V})$~~~ & 12 \\
~Reference output voltage~~ & $x_{4\star}(\text{V})$~~ & -20 \\
~Parasitic Resistance~~ & $r_1~(\text{$\Omega$})$~~ & 1.7 \\
~Parasitic Resistance~~ & $r_2~(\text{$\Omega$})$~~ & 1.7 \\
~Nominal resistance~~ & $r~(\text{$\Omega$})$~~~ & 20\\
~Inductance~~ & $L_1~(m \text{H})$~~~ &  10\\
~Inductance~~ & $L_2~(m \text{H})$~~~ &  10\\
~Capacitance~~ & $C_1~(\mu \text{F})$~~~ & 22\\
~Capacitance~~ & $C_2~(\mu \text{F})$~~~ & 22.9\\
\hline
\end{tabular}
\label{tab:simpars}
\end{table}

\subsubsection*{Performance of the FCT-GPEBO}

First, we evaluate the effect of the {\em adaptation gain} $\gamma$ in the FCT-GPEBO of Proposition \ref{pro2}. We choose $\lambda=5$ and $\mu=10^{-6}$. The initial conditions of the system are $x(0)=(0.75L_1,15C_1,-1.5L_2,-18C_2)$, while those of the FCT-GPEBO are set equal to zero. The transient behavior of the state observation errors, for different adaptation gains $\gamma$, are shown in Figs. \ref{i1_x_known}-\ref{i3_x_known}. As expected, convergence is faster, for larger $\gamma$.

In Fig. \ref{u4_step} we show the transients of the output voltage for two cases:
\begite
\item  a step change in $x_4^\star $ at $t=0.015$ (from $x_4^\star =-15$ to $x_4^\star =-5$),
\item a step change in the load $r$, at $t=0.015$ (from $r=20$ to $r=30$).
\endite
As seen from the figure the regulated output follows quite rapidly the change of reference and is highly insensitive to the load change.

\begin{figure}[!ht]
  \centering
	  \includegraphics[width=1\linewidth]{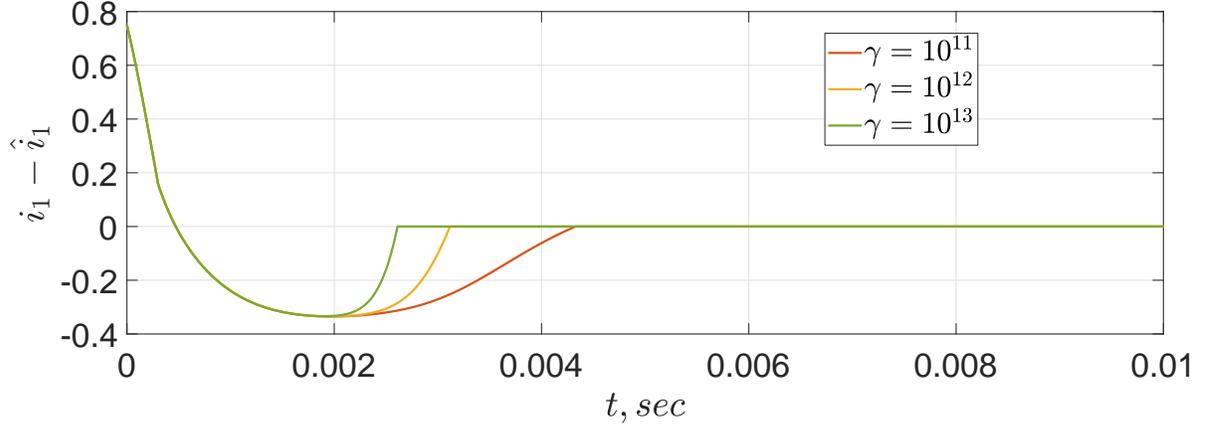}\\
  \caption{Transients of $i_1-\hat i_1$ of the FCT-GPEBO for different adaptation gains}
  \label{i1_x_known}
\end{figure}

\begin{figure}[!ht]
	\centering
	\includegraphics[width=1\linewidth]{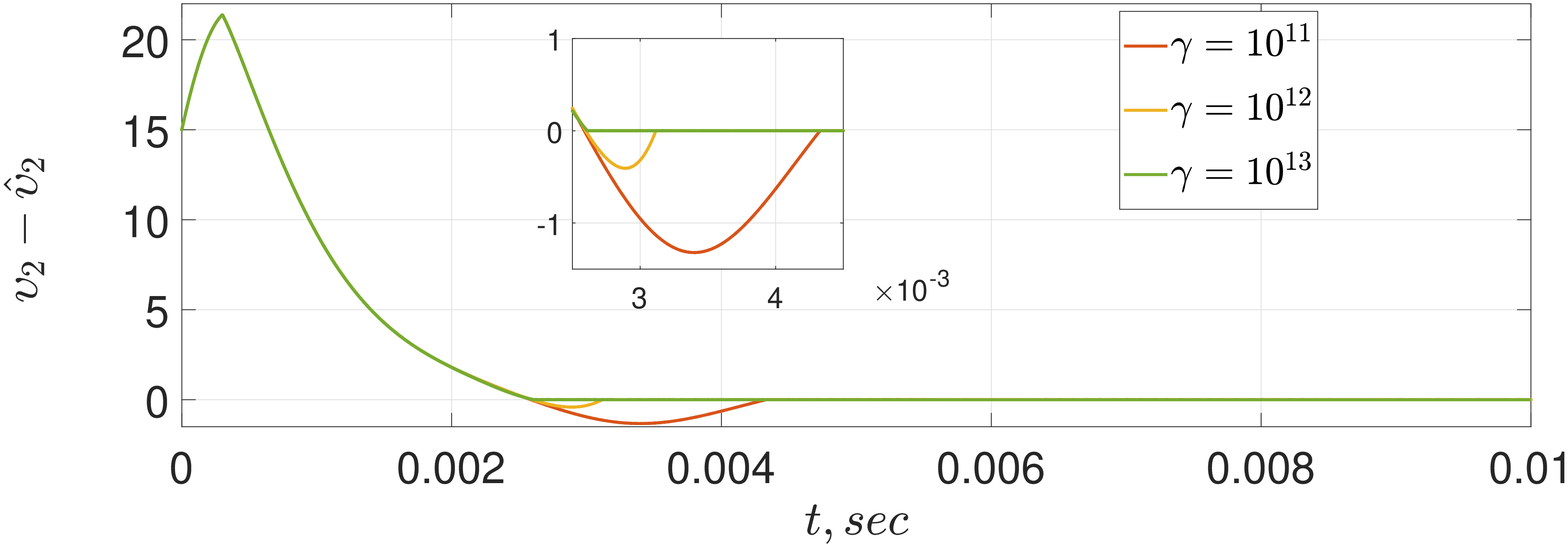}\\
	\caption{Transients of $\upsilon_2-\hat \upsilon_2$ of the FCT-GPEBO for different adaptation gains}
	\label{u2_x_known}
\end{figure}

\begin{figure}[!ht]
	\centering
	\includegraphics[width=1\linewidth]{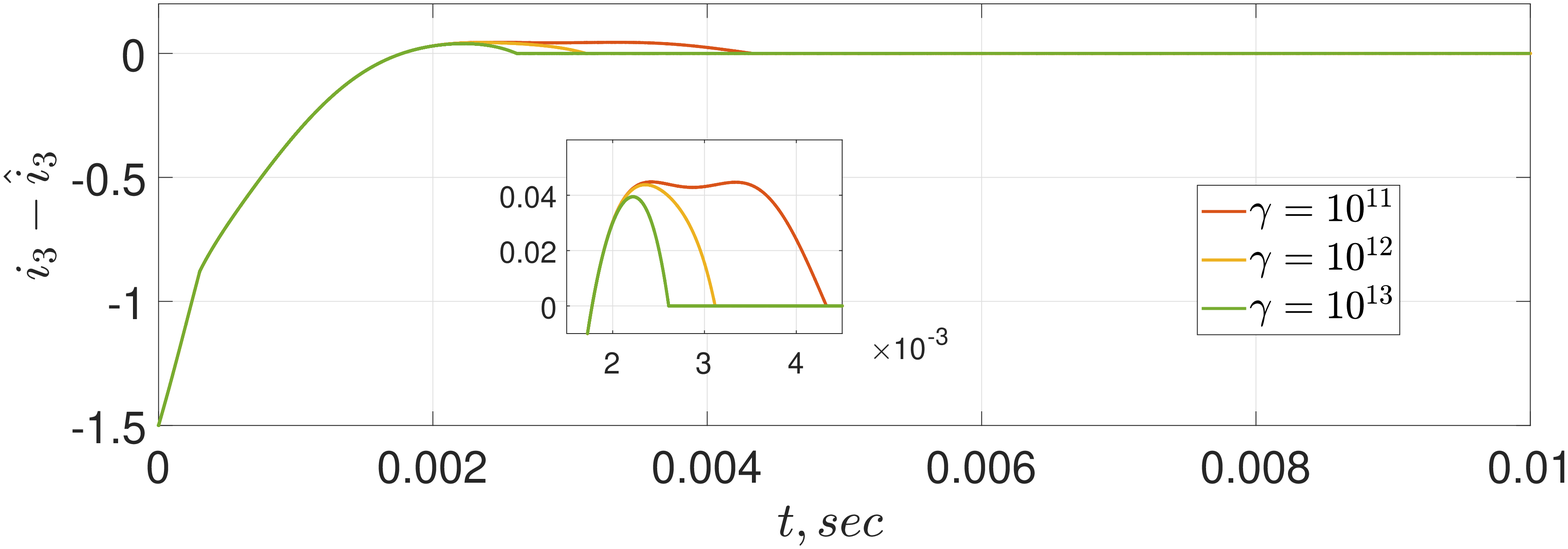}\\
	\caption{Transients of $i_3-\hat i_3$ of the FCT-GPEBO for different adaptation gains}
	\label{i3_x_known}
\end{figure}

\begin{figure}[!ht]
	\centering
	\includegraphics[width=1\linewidth]{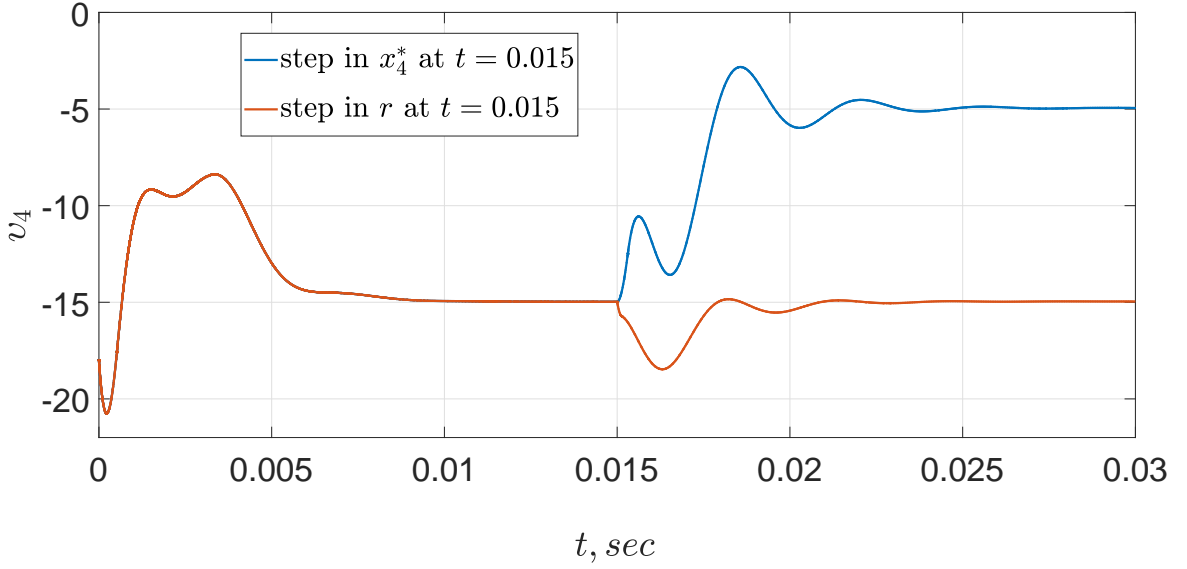}\\
	\caption{Transients of $u_4$ of the FCT-GPEBO for step change in reference signal and load}
	\label{u4_step}
\end{figure}

\subsubsection*{Comparison of the performance of FCT-GPEBO with other observers}

Second, to compare the performance of the FCT-GPEBO of Proposition \ref{pro2} with other observers, we present now simulations of the output-feedback PI-PBC where we replaced this observer by one of the observers discussed in Subsection \ref{subsec32} and the asymptotic GPEBO of Remark \ref{rem5}.  For all observers, the initial conditions of the converter were taken as $x(0)=\col(0.75L_1,15C_1,-1.5L_2,-18C_2)$. For the PEBOs we used $\lambda=5$ and $\gamma=10^{12}$.  For the KBF we used $S=I_{4\times4}$ and  $H(0)=I_{4\times4}$. For the classical gradient estimator we set to zero $\hat{\theta}(0)$ and used $\gamma=10^8$. For the asymptotic GPEBO we used $\gamma=10^{17}$.

Figs.~\ref{i1_comp}-\ref{i3_comp} show the transients of the state estimation errors for the following estimation methods: FCT-GPEBO, asymptotic GPEBO, open-loop emulator, KBF and GPEBO with classical gradient estimator. The following observations are in order.
\begenu[{\bf O1}]
\item As predicted by the theory all state estimation errors of FCT-GPEBO converge in finite time with $t_c \approx 0.03$s, and the performance of this observer is superior to all the other ones.
\item The emulator and the KBF have almost identical behavior. Actually, in the plots the difference is indistinguishable. The performance of the KBF didn't change with other choices of $S$.
\item The asymptotic GPEBO outperforms all other observers, except the FCT-GPEBO.
\item The GPEBO plus gradient also converges, but has an erratic behavior, whose origin we fail to understand.
\endenu

\begin{figure}[h]
	\centering
		\includegraphics[width=1\linewidth]{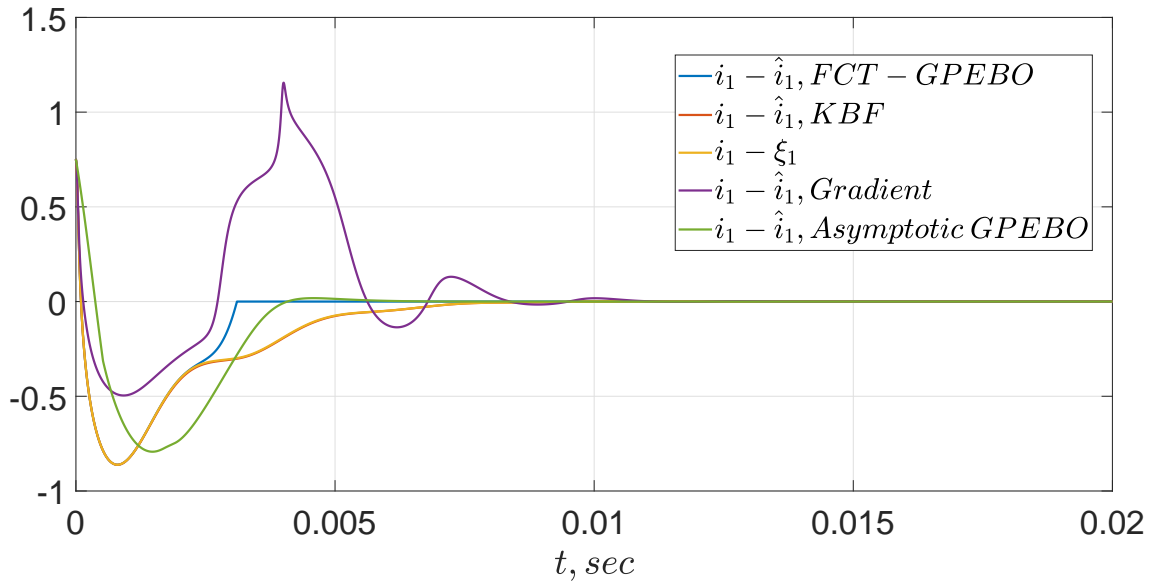}
	\caption{Transients of $i_1$ estimation errors for diffrerent algorithms}
	\label{i1_comp}
\end{figure}

\begin{figure}[h]
	\centering
	\includegraphics[width=1\linewidth]{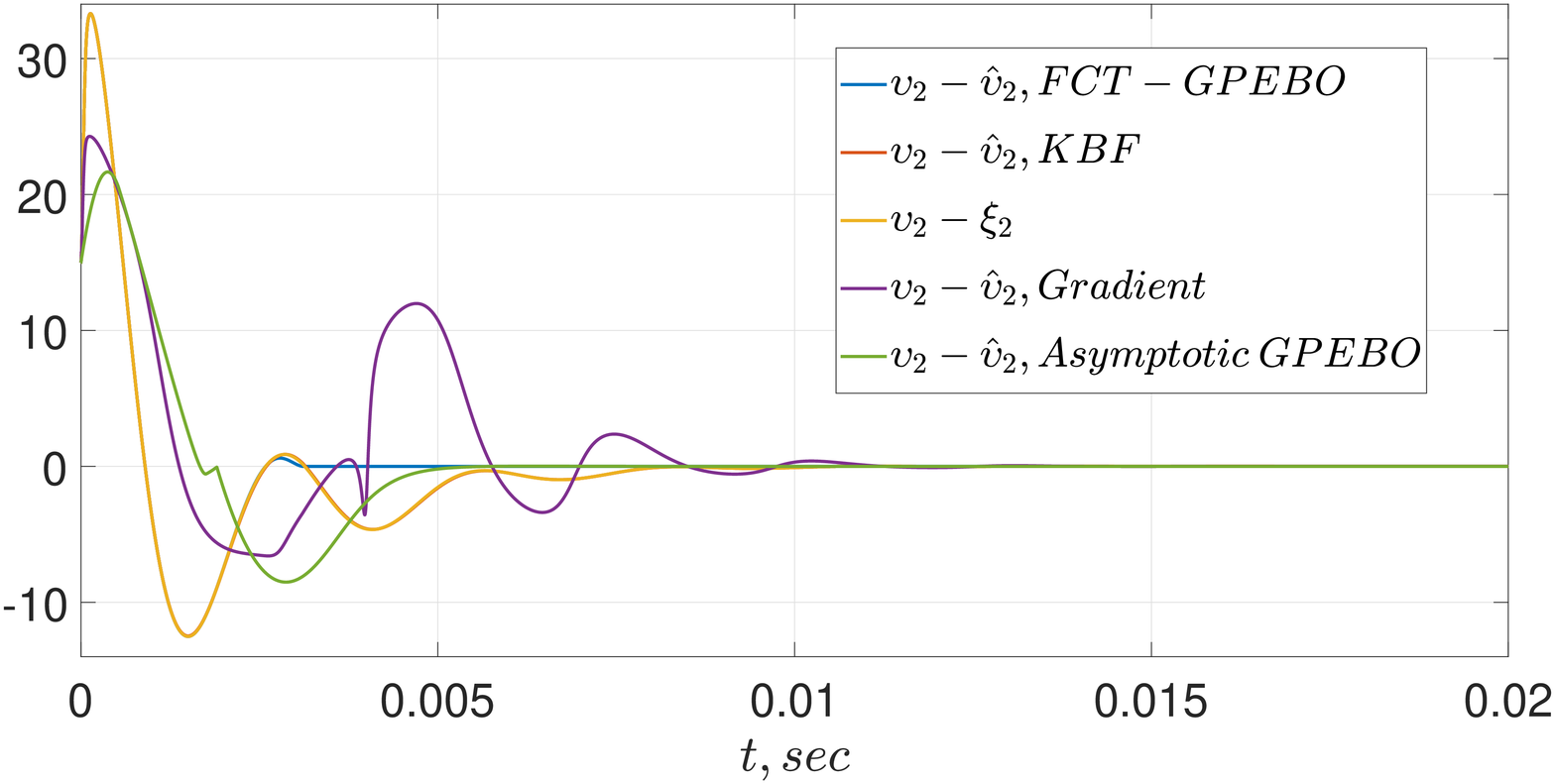}
	\caption{Transients of $\upsilon_2$ estimation errors for diffrerent algorithms}
	\label{u2_comp}
\end{figure}

\begin{figure}[h]
	\centering
	\includegraphics[width=1\linewidth]{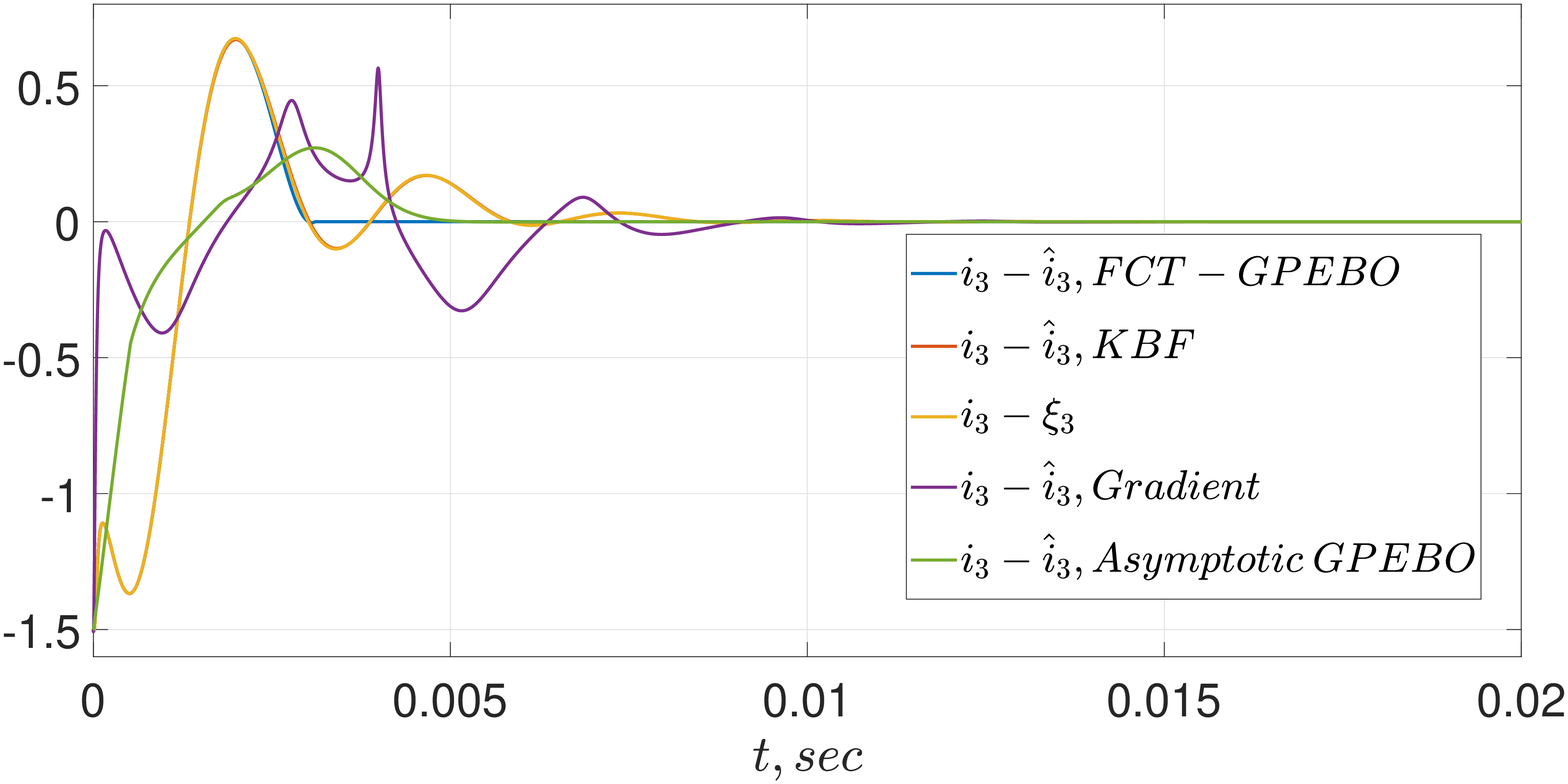}
	\caption{Transients of $i_3$ estimation errors for diffrerent algorithms}
	\label{i3_comp}
\end{figure}
\subsubsection*{Performance of the PI-PBC with full-state measurement and FCT-GPEBO}

Second, we compare the performance of the full-state feedback PI-PBC with the one using the observed state. We look in this case at the effect of the {\em initial conditions}. For simulations we considered the initial conditions:
\begalis{
x(0) &=\col(0.5L_1,10C_1,-1L_2, -12C_2)\\
x(0)&=\col(0.25L_1,5C_1,-0.5L_2,-6C_2)\\
x(0)&=\col(0.75L_1,15C_1,-1.5L_2,-18C_2),
}
and in FCT-GPEBO we chose the parameters $\lambda=5$, $\gamma=10^{12}$. The behavior of the output voltage for PI-PBC with the {\em known} state is shown in Fig.~\ref{u4_x_known}, while the one of the observed state is shown in Fig. \ref{u4_x_unknown}. As seen from the figures the reconstruction of the state induces some oscillations at the beginning but the regulation of the output is achieved at almost the same time.

\begin{figure}[h]
	\centering
	\includegraphics[width=1\linewidth]{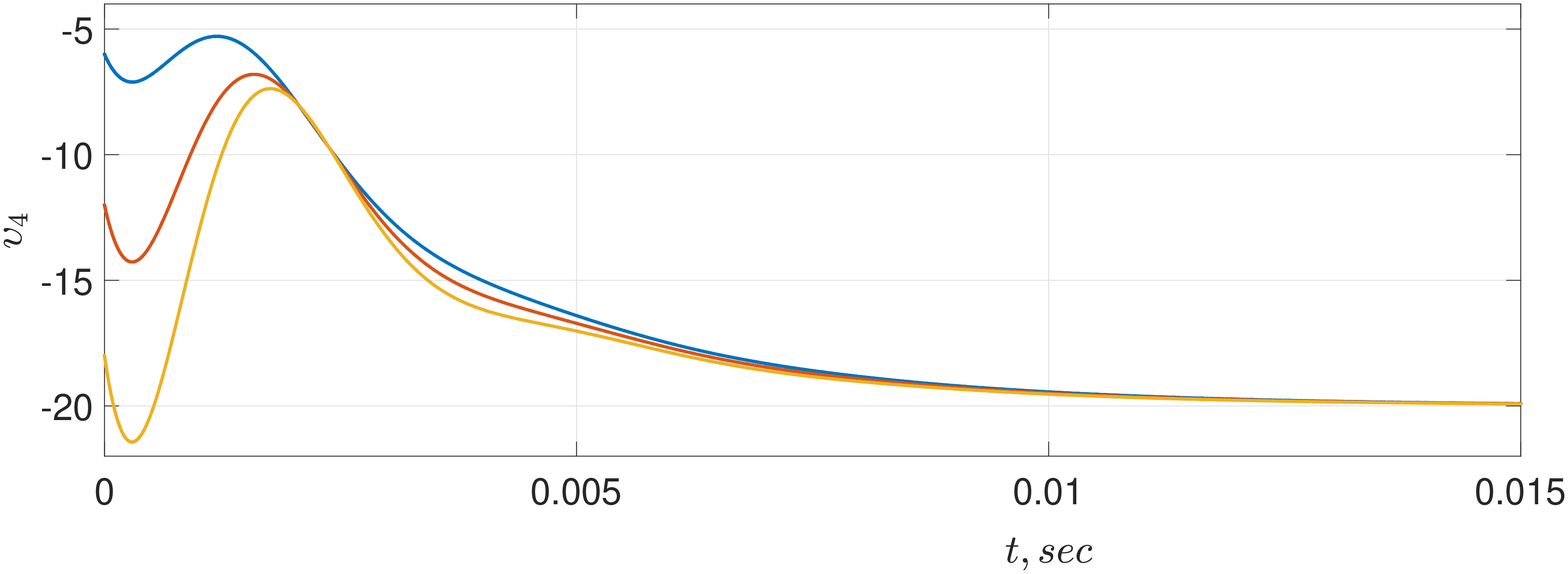}
	\caption{Transients of $\upsilon_4$ for different initial conditions and the PI-PBC with {\em known} $x$}
	\label{u4_x_known}
\end{figure}

\begin{figure}[h]
	\centering
	\includegraphics[width=1\linewidth]{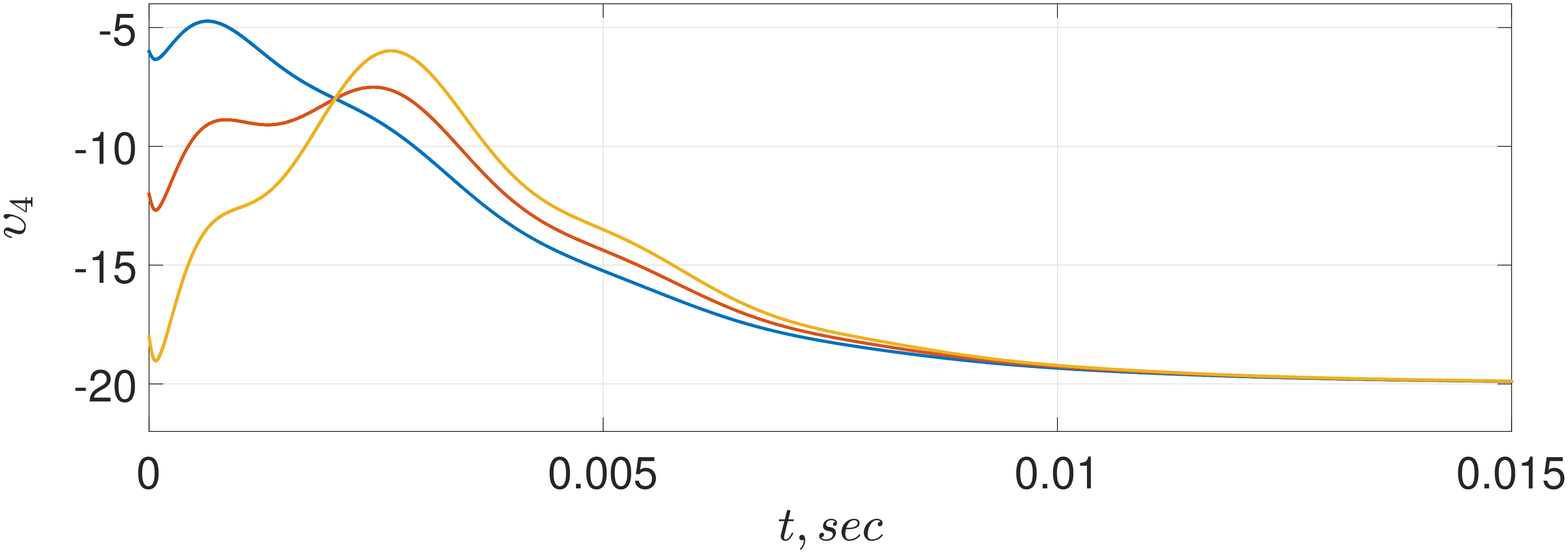}
	\caption{Transients of $\upsilon_4$ for different initial conditions and the PI-PBC with {\em unknown} $x$}
	\label{u4_x_unknown}
\end{figure}

Finally, we compare the output-feedback PI-PBC with a {\em classical} PI wrapped around the output voltage error, that is
\begin{align*}
\dot x_c &= x^\star _4 - x_4 \\
u & = -K_P(x^\star _4 - x_4)-K_I x_c.
\end{align*}
The results are shown in Fig. \ref{fig5}, showing the superior performance of the PI-PBC.

\begin{figure}[!ht]
  \centering
\includegraphics[width=1\linewidth]{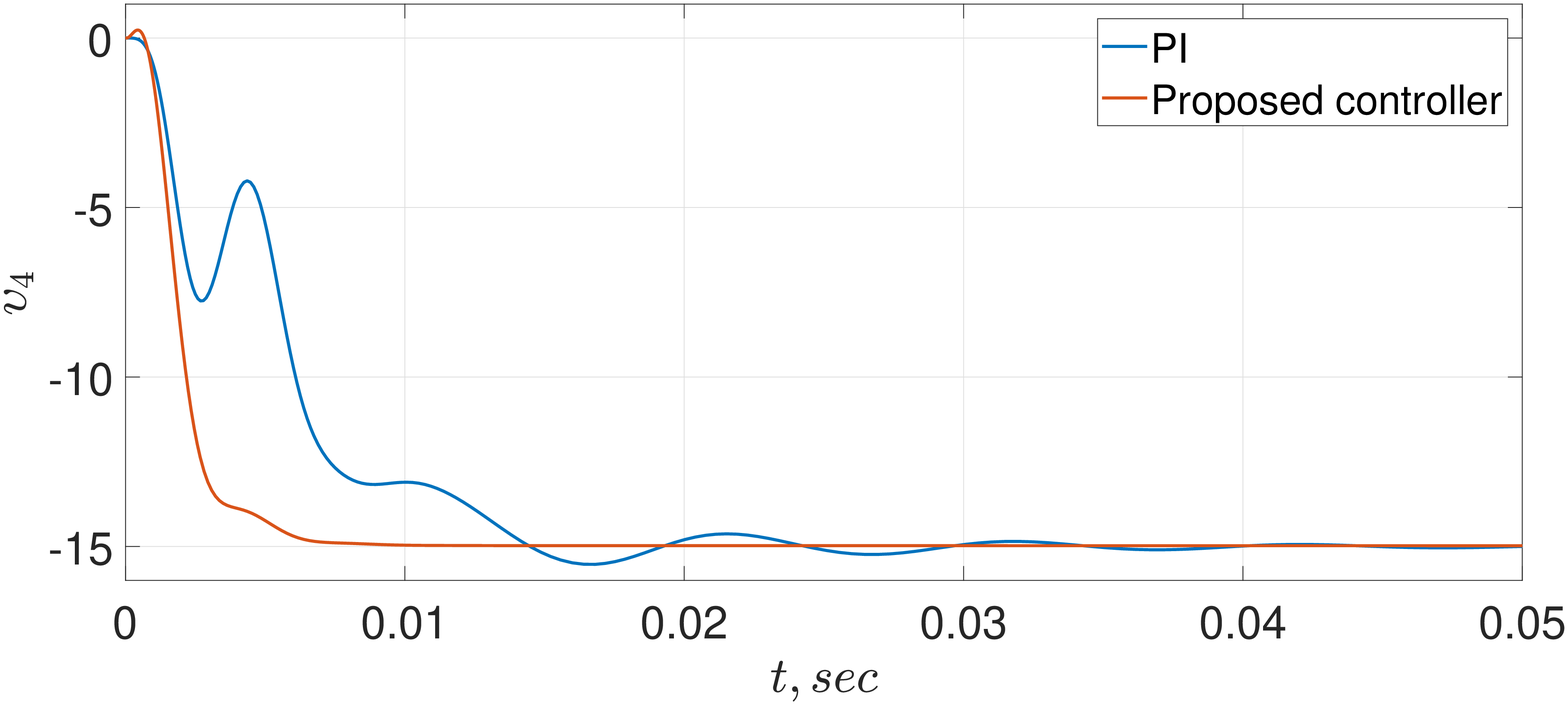}\\
  \caption{The response curves of $x_4$ of the closed-loop system under traditional PI with $K_P=0.008, K_I=8$ and the output-feedback PI-PBC, with $x_4^\star =-15$}
  \label{fig5}
\end{figure}

It may be argued that the response of the classical PI can be improved re-tuning the gains.  The usual procedure to improve the transient performance of a PI is to increase the integral gain $K_I$, see \cite{WANetalbook}. This is shown in Fig. \ref{fig6} where the transient performance is, indeed, improved but is still inferior to the proposed PI-PBC.

\begin{figure}[!ht]
  \centering
  \includegraphics[width=1\linewidth]{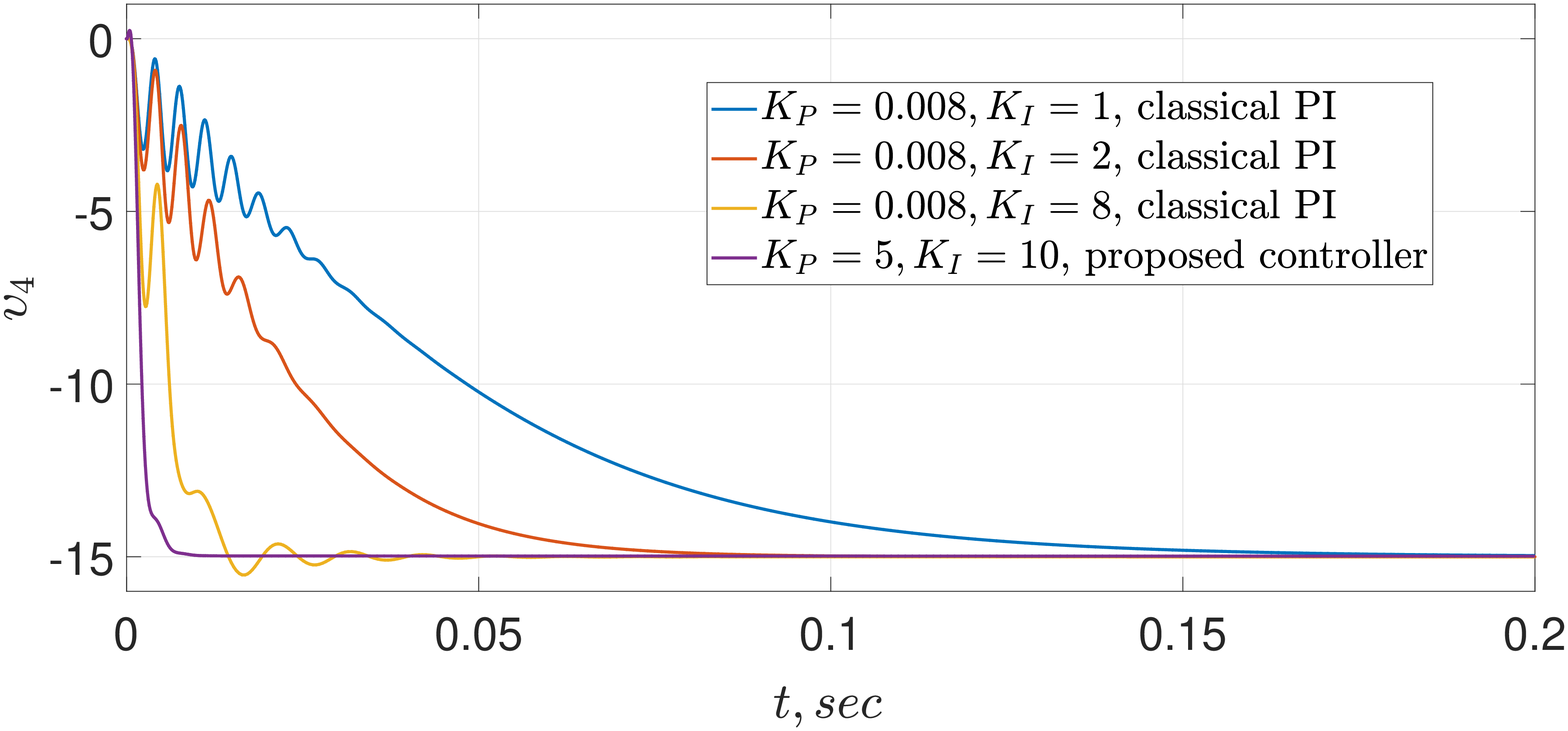}\\
  \caption{The response curves of $x_4$ of the closed-loop system under traditional PI with $K_P=0.008$ and different values of $K_I$, with $x_4^\star =-15$}
  \label{fig6}
\end{figure}

%
\section{Conclusions an Future Work}
\lab{sec6}
%
We have provided a solution to the main drawback of the popular PI-PBC of \cite{HERetal}, namely the need to measure the full state, by proposing a globally convergent state observer. The convergence of this observer occurs in finite time and requires only a very weak excitation assumption.

In view of the high uncertainty on the converter parameters---in particular, the variations of the load---our current research focuses on the development of adaptive versions of the observer and the PI-PBC. Some results for particular converter topologies are already available and will be reported elsewhere.

We are also working on the experimental implementation of the proposed PI-PBC for a \'{C}uk and a quadratic converter, that we expect to have ready for the final version of this paper.


\end{document}